\documentclass{article}

\usepackage[nonatbib,final]{neurips_2024}

\usepackage[utf8]{inputenc} 
\usepackage[T1]{fontenc}    
\usepackage{hyperref}       
\usepackage{url}            
\usepackage{booktabs}       
\usepackage{amsfonts}       
\usepackage{nicefrac}       
\usepackage{microtype}      
\usepackage{xcolor}         

\usepackage{morenotations}

\usepackage{booktabs}
\usepackage{multirow}
\usepackage[normalem]{ulem}
\useunder{\uline}{\ul}{}

\usepackage{amsmath}
\usepackage{amssymb}
\usepackage{mathtools}
\usepackage{amsthm}
\usepackage{mathrsfs}
\usepackage{subcaption}

\usepackage[textsize=tiny]{todonotes}
\theoremstyle{plain}
\newtheorem{theorem}{Theorem}[section]
\newtheorem{proposition}[theorem]{Proposition}

\theoremstyle{definition}
\newtheorem{definition}[theorem]{Definition}

\theoremstyle{remark}

\title{\papertitle}

\author{%
  Hisham Husain$^{*}$
  \And
  Julien Monteil$^{+}$
  }
  
\begin{document}
 





\maketitle

\begin{abstract}
Latent variable collaborative filtering methods have been a standard approach to modelling user-click interactions due to their simplicity and effectiveness. However, there is limited work on analyzing the mathematical properties of these methods in particular on preventing the overfitting towards the identity, and such methods typically utilize loss functions that overlook the geometry between items. In this work, we introduce a notion of generalization gap in collaborative filtering and analyze this with respect to latent collaborative filtering models. We present a geometric upper bound that gives rise to loss functions, and a way to meaningfully utilize the geometry of item-metadata to improve recommendations. We show how these losses can be minimized and gives the recipe to a new latent collaborative filtering algorithm, which we refer to as GeoCF, due to the geometric nature of our results. We then show experimentally that our proposed GeoCF algorithm can outperform other all existing methods on the Movielens20M and Netflix datasets, as well as two large-scale internal datasets. In summary, our work proposes a theoretically sound method which paves a way to better understand generalization of collaborative filtering at large.
\end{abstract}

\section{Introduction}\label{sec:introduction}
The importance of recommender systems cannot be overstated and play a pivotal role in many e-commerce and entertainment applications. The general problem of recommendation is to consume user-item interactions, and learn meaningful relationships to better recommend items that the user has not seen, but is likely to enjoy. A number of different algorithmic solutions, from neighbourhood methods, e.g.~\cite{koren2008factorization}, linear regressions, e.g.~\cite{ning2011slim}, factorization models, e.g.~\cite{rendle2022revisiting}, to sequential models, e.g.~\cite{wu2020sse}, and latent generative models, e.g.~\cite{shenbin2020recvae}, have been proposed throughout the years to solve this practical problem. Significant effort has been dedicated to surpass the performance of state-of-the-art methods on public datasets, focusing on learning the best representation of the sparse binarized user-item interaction matrix. Less interest has been vested on leveraging item and user metadata information to boost the learning capabilities of the algorithms, with most existing approaches leveraging regularizers to ground the learning of the latent variables, e.g.~\cite{ning2012sparse}, in particular in the context of cold-start recommendations, e.g.~\cite{wang2023equivariant}.

One particular top-performing method involves latent collaborative filtering, which seeks to reconstruct the user-item interaction with the use of a generative process. The assumption behind the success that generative models have in reconstructing the click history may be that the stochastic process helps in preventing the overfitting towards the identity, which is considered to be the principal cause behind the loss of predictive capability~\cite{steck2020autoencoders}. One popular approach in this family is the MultiVAE \cite{liang2018variational}, and its subsequent variations~\cite{vanvcura2021deep,kim2019enhancing} which dominate the MovieLens20M and Netflix leaderboards, allow for flexible architectures, and are efficient. Despite the success, there are two main caveats. Firstly, the majority of these approaches are agnostic to item metadata and second, there is little theoretical understanding of these approaches from a generalization perspective. While we discuss these as two separate problems, they form the underpinning of a more compelling question we seek to answer: \textit{Can we derive a latent collaborative filtering approach that achieves theoretical guarantees on generalization via item metadata?}

In order to answer this question, we first need to define a notion of generalization for latent collaborative filtering. The closest concept to generalization has been commonly recollected as cold-start \cite{volkovs2017dropoutnet}, and while the literature on the topic has become denser, the endeavours have been largely empirical, see e.g.~\cite{pan2019warm,xu2022alleviating,monteil2024marec}. In Statistical Learning Theory (SLT)~\cite{bartlett2002rademacher}, the motivation behind generalization is to counter overfitting. Despite the prevalence of overfitting in collaborative filtering~\cite{steck2020autoencoders}, there is surprisingly little work in the formalization of generalization. In this work, we borrow machinery from SLT and define, to the best of our knowledge, the first formalization of a generalization gap in latent collaborative filtering. We show that our definition aligns well and parallels existing intuition in latent model CF such as the trepidation of overfitting to the identity model. 

We then present a geometrically dependent upper bound on the generalization gap for an arbitrary latent collaborative filtering method, that gives rise to loss functions. We show how these losses can be minimized and gives the recipe to a new latent collaborative filtering algorithm, which we call GeoCF, due to the geometric nature of our results. 

In our geometrical upper bound, we find a convergence rate of $n^{-\frac{1}{d^{*}}}$ ($n$ being the number of observed data points) where $d^{*}$ is the \textit{intrinsic} dimensionality of our user-item distribution induced by the item metadata. The quantity $d^{*}$ depends directly on how well the item space is metrized by embeddings, thus emphasizing the importance item metadata plays for performance from a theoretical perspective. Finally, we show that GeoCF empirically outperforms existing baselines on popular baselines such as Movielens20M and Netflix datasets.

In summary, our contributions are three-fold:
\begin{itemize}
    \item A formalization of generalization gap in latent collaborative filtering methods, which parallels and subsumes our intuition about existing methods.
    \item (GeoCF) A new latent collaborative filtering algorithm motivated directly from the geometric generalization bound that utilizes item meta-data for recommendation.
    \item Convergence result for GeoCF, establishing the importance of item metadata in improving generalization from a formal perspective. 
\end{itemize}

\section{Preliminaries on Collaborative Filtering}
In this section, we will first introduce preliminaries for latent collaborative filtering models followed by a brief summary of the optimal transport tools to be used in building the proposed algorithm.

\paragraph{Collaborative Filtering} We use $\Omega$ to denote a compact Polish space and denote $\Sigma$ as the standard Borel $\sigma$-algebra on $\Omega$, $\mathbb{R}$  denotes the real numbers and $\mathbb{N}$ natural numbers. Let $I \in \mathbb{N}$ denote the number of all items and we will consider $\Omega = \mathbb{R}^I$ to represent users where each $X \in \Omega$ corresponds to the click-vector of a specific user $X$. For simplicity, we will binarize this vector: for item $i$, user $X$ clicked on it if $X^{i} = 1$ and $X^{i} = 0$ otherwise. We use $\Delta(\Omega)$ to denote the set of all probability measures over $\Omega$. 

We observe a set of item-user click history consisting of $n$ users, which we denote by $\mathsf{X} = \braces{X_i}_{i=1}^n$ where each $X_i \in \Omega$. The goal of collaborative filtering is to recommend items to users by first identifying similar users based on their history. For example, if the history of two users $X_i$ and $X_j$ differ only by one click, then it is likely that item is of interest to both users. Thus, we would like to produce a function $F$ that takes in a user $X \in \Omega$ and outputs a item-click vector $\tilde{X} \in \Omega$ that the user would most likely be interested in. Thus, we can formulate this as a least regression problem:
\begin{align}
    \min_{F : F \neq \operatorname{Id}} \nrm{\mathsf{X} -  \mathsf{X}F^\top},
\end{align}
where we treat $\mathsf{X} \in \mathbb{R}^{n \times I}, F \in \mathbb{R}^{I \times I}$. We want to avoid the identity solution, and this is often enforced in the form of a regularizer, e.g.~\cite{steck2019embarrassingly}. Various developments in this line of work study how to decompose $F$ to incorporate such intuition. A straightforward way to perform collaborative filtering is via neighbourhood methods, which consist of 
	computing similarities across items or users, where the similarity is evaluated as a dot product on click vectors, and possibly 
	item and user metadata~\cite{4781121}. The Netflix competition and subsequent research works on available public datasets have shown that SLIM~\cite{ning2011slim} and matrix factorization techniques tend to 
	outperform neighbourhood methods~\cite{koren2009matrix,steck2020admm}. A push for neural approximations 
	of matrix factorization took place~\cite{he2017neural,covington2016deep}, with the conclusion that the proposed 
	MLP architectures fail to learn a better nonlinear variant of the dot product and are outperformed by careful implementations of  
	matrix factorization~\cite{dacrema2019we,rendle2020neural,rendle2022revisiting}.  More recently, some neural approaches have obtained state of the art results. The variational autoencoder approach~\cite{liang2018variational} (MultiVAE) is amongst the most competitive neural approaches but the least squares approximation of SLIM, i.e., a linear autoencoder with projections to ensure the zero diagonal constraint (EASE)~\cite{steck2019embarrassingly,steck2020admm}, still beat it on 2 out 3 public datasets. It was shown in~\cite{kim2019enhancing} that augmenting the MultiVAE approach with flexible priors and gating mechanisms led to state of the art performance. The MultiVAE approach was also refined and ensembled with a neural EASE to beat all baselines on MovieLens20M and Netflix~\cite{vanvcura2021deep}. The success of Deep Cross Networks on learning to rank tasks~\cite{qin2021neural} also showed the importance of explicitly encoding the dot product in the network architecture. Finally, a more recent approach related to the proposed framework in this paper is SinkhornCF \cite{li2021sinkhorn}, which replaces the reconstruction error by the Sinkhorn distance, a loss that will help guide the model to recommend items based on their metadata. We present the preliminaries of the Sinkhorn loss and optimal transport in the next section. 

\paragraph{Optimal Transport and Divergences}
When it comes to comparing two probability distributions $P,Q \in \Delta(\Omega)$, a very well regarded and canonical choice is the Wasserstein distance \cite{villani2009optimal}, which is defined using a ground cost $c: \mathcal{X} \times \mathcal{X} \to \mathbb{R}_{+}$
\begin{align}
\mathscr{W}_c(P,Q) := \inf_{\pi \in \Pi(P,Q)} \braces{\int_{\mathcal{X} \times \mathcal{X}} c(x,y) d\pi(x,y) },
\end{align}
where $\Pi(P,Q)$ is the \emph{set of couplings} between $P$ and $Q$
\begin{equation}
\Pi(P,Q) = \biggl\{ \pi \in \mathscr{P}(\mathcal{X} \times \mathcal{X}) :  \int_{\mathcal{X}} \pi(x,y) dx = P, \int_{\mathcal{X}} \pi(x,y) dy = Q \biggl\}.
\end{equation}
Finding the exact value of the Wasserstein distance is expensive and one often resorts to regularizing the objective with the entropy of $\pi$, yielding the following objective:
\begin{equation}
    S^{\varepsilon}(P,Q) = \inf_{\pi \in \Pi(P,Q)} \biggl\{\inf_{\mathcal{X} \times \mathcal{X}}c(x,y) d\pi(x,y) + \varepsilon \int_{\mathcal{X} \times \mathcal{X}} \log \pi(x,y) d\pi(x,y)  \biggl\}.
\end{equation}
This objective can be solved in a more stable and tractable manner, referred to as the Sinkhorn algorithm \cite{cuturi2013sinkhorn}. Since $S^{\varepsilon}$ serves as a viable approximation to $\mathscr{W}_c$, it has been popularly used in various machine learning applications. Another choice of divergence we will utilize is the Maximum Mean Discrepancy (MMD) with respect to a kernel $k$ $d_{\operatorname{MMD}}(P,Q)$ which is typically used for its computational convenience.
\section{Generalization for Collaborative Filtering}
\label{sec:generalization}
In this section, we present to the best of our knowledge, the first formal pursuit of a generalization study for latent collaborative filtering. We begin by introducing a notion of generalization error, akin to the generalization gap existing in supervised learning \cite{bartlett2002rademacher}. Letting $2^{\Omega}$ refer to the power set of $\Omega$, we represent a collaborative filtering algorithm as a mapping $C: \Omega \times 2^{\Omega} \to \Delta(\Omega)$ so that $C(\X, \mathcal{D})$ represents the distribution of possible recommendations for user $\X$ upon training with data $\mathcal{D}$.
\begin{definition}[Collaborative Filtering Generalization Error]
Let $P \in \Delta(\Omega)$, let $C: \Omega \times 2^{\Omega} \to \Delta(\Omega)$ denote a collaborative filtering algorithm, then the error of $C$ under $n$ samples is defined as \begin{align}\mathcal{E}_{n,P}(C) := \E_{X^n \sim P^n}\left[ d_{\operatorname{TV}}\bracket{\int_{\mathcal{X}} C(\X,X^n) dP(\X) , P}\right],\end{align}
where $d_{\operatorname{TV}}(\mu,\nu) = \sup_{A \in \Sigma} \bracket{\mu(A) - \nu(A)}$, is the total-variation distance between distributions $\mu,\nu \in \Delta(\Omega)$.\label{def}
\end{definition}
We are interested in analyzing how well the collaborative filtering algorithm $C$ consumes a dataset $X^n$ of size $n$ from $P$ to reconstruct the user distribution (and their preferences). Thus, we analyze the discrepancy using the total variation metric, and take the expectation across all datasets of size $n$.
Note that the quantity $\mathcal{E}_{n,P} \in [0,1]$ and thus we want to make sure it decreases to $0$ at a fast rate as $n \to \infty$. We now show that the identity solution characterized as $C_{\operatorname{ID}}(\X,\mathcal{D}) = \frac{1}{\card{\mathcal{D}}}\sum_{\X' \in \mathcal{D}} \delta_{\X'}(\X)$ that overfits to the data achieves $1$, which is the maximal possible value.
\begin{proposition}
Let $P \in \Delta(\Omega)$ then $\mathcal{E}_{n,P}(C_{\operatorname{ID}}) = 1$ for any $n < \infty$.
\end{proposition}
\begin{proof}
    Note that the distribution $C_{\operatorname{ID}}$ is not absolutely continuous with respect to $P$ and thus the Total Variation Distance will be $1$ with probability $1$.
    \label{prop1}
\end{proof}
While this is not a surprise and anticipated, as discussed in \cite{steck2020autoencoders}, we would like to emphasize this is the first negative result connecting to a formal notion of generalization. In the proof, we can see that the negative result stems from the fact that the model is unable to recommend for unseen users and so we want the resulting distribution $C(\X, \mathcal{D})$ to exemplify some smoothness. We now analyze the behaviour of $\mathcal{E}_{n,P}$ for latent collaborative filtering algorithms.

In such models, we have a latent space $\mathcal{Z}$ with a prior distribution $\gamma(z) \in \Delta(\mathcal{Z})$. We use $\upvartheta$ and $\xi$ to denote the parameters of a decoder and encoder distribution $p_{\upvartheta}(\X \mid z)$ and $q_{\xi}(z \mid \X)$ respectively. We will also assume that we have a deterministic decoder so that $p_{\upvartheta}(X\mid z) = \delta_{D_{\upvartheta}(z)}(X)$ where $D_{\upvartheta}: \mathcal{Z} \to \mathcal{X}$. During training time, we minimize a loss $\mathsf{L}(\upvartheta,\xi)$ and the resulting collaborative filtering algorithm for recommendation on a user $\X$ is
\begin{align}
    C_{\upvartheta,\xi}(\X,\mathcal{D}) = \int_{\mathcal{Z}} p_{\upvartheta}(\cdot \mid z) q_{\xi}(z \mid \X)dz.
\end{align}
The above describes the reconstruction process given by the encoder and decoder, which we expect will be much smoother than in the case of the dirac delta, as seen in the identity case. Our goal is to provide an analysis of $\mathcal{E}_{n,P}\bracket{C_{\upvartheta,\xi}}$ and we do so by exploiting the geometry of the space in $\Omega$. To this end, assume we have a cost $c: \Omega \times \Omega \to \mathbb{R}$ that metrizes $\Omega$, which allows us to define the following machinery from \cite{weed2019sharp}.
\begin{definition}[$c$-Covering Numbers]
For a set $S \subseteq \Omega$, we denote $N_{\eta}(S)$ to be the $\eta$\emph{-covering number} of $S$, which is the smallest $m \in \mathbb{N}_{*}$ such that there exists closed balls $B_1,\ldots,B_m$ of radius $\eta$ with $S \subseteq \bigcup_{i=1}^m B_i$. For any $P \in \Delta(\Omega)$, the $(\eta, \tau)$\emph{-dimension} is $d_{\eta}(P, \tau) := \frac{\log N_{\eta}(P, \tau) }{-\log \eta},$ where $N_{\eta}(P,\tau) := \inf\braces{N_{\eta}(S) : P(S) \geq 1 - \tau}.$
\end{definition}
\begin{definition}[$1$-Upper $c$-Wasserstein Dimension]
The $1$\emph{-Upper Wasserstein dimension} of any $P \in \Delta(\Omega)$ is $d^{*}_{c}(P) := \inf\braces{s \in (2, \infty) : \limsup_{\eta \to 0} d_{\eta} (P, \eta^{\frac{s}{s -  2}}) \leq s }.$
\end{definition}

The quantity $d^{*}_c(P)$ represents the \textit{intrinsic} dimension of where $P$ is largely supported. For example, if $P$ lies on a low dimensional manifold in some high dimensional ambient space, then $d^{*}_c(P)$ would correspond to the dimension of said manifold. We now present an upper bound on $\mathcal{E}_{n,P}(C_{\upvartheta, \xi})$.
\begin{theorem}
Let $\mathcal{F}_c$ denote the set of $1$-Lipschitz functions with respect to $c$ and $\mathcal{H}_k$ denote the unit ball of Reproducing Kernel Hilbert Space functions with kernel $k$. For any $P \in \Delta(\Omega)$ and $\lambda > 0$, consider the minimizers $
    \upvartheta^{*}, \xi^{*} = \arginf_{\upvartheta,\xi} \mathsf{L}(\upvartheta, \xi).$
It then holds that with high probability
\begin{equation}
    \mathcal{E}_{n,P}(C_{\upvartheta^{*}, \xi^{*}}) \leq \delta(\upvartheta^{*}, \xi^{*}; P) + \frac{1}{4} \E \biggl[ \mathscr{W}_c\bracket{\hat{P}^n, p_{\upvartheta^{*}} \# \gamma}
    + d_{\operatorname{MMD}}\bracket{q_{\xi^{*}} \# 
    \hat{P}^n, \gamma}\biggr] + An^{-\frac{1}{d^{*}_c(P)}},
\end{equation}
for some $A > 0$ where $q_{\xi^{*}}\# \hat{P}^n (z)= \int_{\mathcal{X}} q_{\xi^{*}}(z \mid \X) d\hat{P}^n(\X)$, $p_{\upvartheta^{*}}\# \gamma = \int_{\mathcal{Z}} p_{\upvartheta^{*}}(\cdot \mid z) \gamma(z)$ and $\delta(\upvartheta^{*}, \xi^{*}; P) = 2 \inf_{h \in \mathcal{F}_c} \nrm{\log \frac{P}{p_{\upvartheta^{*}} \# \gamma} - h}_{\infty} + 2\inf_{h \in \mathcal{H}_k}\nrm{\log \frac{\gamma}{q_{\xi^{*}} \# P} - h}_{\infty}$.
\label{main_result}
\end{theorem}
\begin{proof}(Sketch, full proof in Appendix)
We apply a triangle inequality to the Total Variation (TV) distance in $\mathcal{E}_{n,P}$ over the latent space. We then utilize two applications of Pinskers inequality to relate the TV distances to the Kullback-Leibler divergence and Integral Probabiltiy Metrics (IPMs), whose dual forms correspond to the Wasserstein distance and the MMD. Applying concentration inequalities on these IPMs allows us to get fast rates of convergence for MMD and a dimensionality dependence when applied to the Wasserstein distance, resulting in the final inequality.
\end{proof}

The proof of this Theorem can be found in the Appendix. The above result tells us that the generalization gap is bounded by three terms. The first term is independent of $n$ and depends on how well $P$ and $\gamma$ can be approximated by the encoder and decoder distributions. In particular, noting that $\mathcal{H}_k$ are universal function approximators, these terms can be kept arbitrarily small. The second term is the sum of two expected gaps: the model $p_{\upvartheta^{*}}$ and $\hat{P}^n$, along with the gap between the prior distribution $\gamma$ and the encoder $q_{\xi^{*}}$.

\vspace{1cm}
\begin{figure}[h!]
    \centering
    \includegraphics[width=0.5\textwidth]{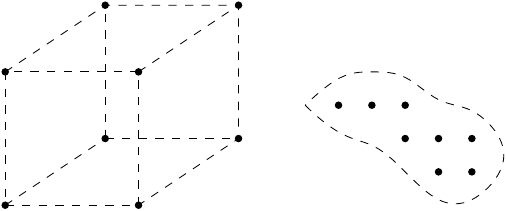}
    \caption{An illustration of the distribution $P$ embedded in the space $\Omega$ for two different choices of $c_I$. The left corresponds to when users are more spread out whereas the right corresponds to a more meaningful representation so that the distribution $P$ is supported on a lower dimensional manifold.}
    \label{fig:intrinsic}
\end{figure}
\vspace{1cm}

The third term is the rate of convergence, which depends on the complexity of our distribution $P$ and choice of metric $c_I$. To better understand this, observe that $d^{*}_c(P)$ is a measure of how well the support of $P$ can be covered by $\eta$-radius balls, where this radius is measured by $d(u,u') = \mathscr{W}_{c_{I}}(p_u, p_{u'})$. We illustrate in Figure \ref{fig:intrinsic} the user space $\Omega$ can change drastically depending on the metric $c_I$ used. 

We will now derive an algorithm in the next section inspired by the upper bound.

\section{GeoCF: utilizing geometry to close the generalization gap}\label{geocf}

We saw in the previous section the generalization gap of collaborative filtering can be lowered by controlling parameters $\upvartheta, \xi$ and metric $c$ of the upper bound:
\begin{equation}
      \underbrace{\E\left[\mathscr{W}_c\bracket{\hat{P}^n, p_{\upvartheta} \# \gamma}\right]}_{\mathclap{\text{Decoder}}}
    + \underbrace{\E\left[d_{\operatorname{MMD}}\bracket{q_{\xi} \# 
    \hat{P}^n, \gamma}\right]}_{\mathclap{\text{Encoder}}} + \underbrace{An^{-\frac{1}{d^{*}_c(P)}}}_{\mathclap{\text{Metric}}}.
\end{equation}
Schematically, we can derive an algorithm by first picking a choice of $c$ that helps decompose the user distribution $P$ into a more compact representation and secondly, minimizing the two terms with respect to $\upvartheta$ and $\xi$.
\paragraph{Metrizing $\Omega$ with $c$} Based on the theoretical result, we need to metrize $\Omega$ so that $d_{c}^{*}$ is reduced. Intuitively, this means we must find a way to separate users so that the natural distribution $P$ is compactly preserved, and for this task we rely on item metadata. Suppose we have an embedding function $E: \braces{1,\ldots,I} \to \mathbb{R}^k$ that embeds each of the $I$ items into a representation into $\mathbb{R}^k$. By considering a user $u \in \Omega$ as point-cloud distribution consisting of items they clicked on: $p_u = \frac{1}{\card{I_u}}\sum_{i \in I_u} \delta_{i}$ where $I_u$ is the set of items clicked on by user $u$, we can define a distance between users: $d_I(u,u') = \mathscr{W}_{c_I}(p_u,p_{u'})$
where $c_{I}(i, i') = \nrm{E(i) - E(i')}_2$ is used as the ground cost between items. The distance $d_I$ compares users by considering the semantic similarity of their respectively engaged items, and since our hypothesis is that $P$ naturally clusters users with shared interests together, it is natural that $d^{*}_{d_I}(P)$ results in a smaller quantity. In the extreme opposite case where there are no embeddings or $01$-loss is used, for example in MultiVAE (see \cite{li2021sinkhorn}[Section~2.3]), then the dimensionality of the support of $P$ will be higher which inturn leads to slow convergence.

\paragraph{Learning $\xi$} For the encoder parameters $\xi$, the most straight forward way is to minimize the term $d_{\operatorname{MMD}}(q_{\xi} \# \hat{P}^n, \gamma)$ which can indeed be computed in closed form with $O(m^2)$ sample complexity if we take $m$ samples from $\gamma$ and $q_{\xi} \# \hat{P}^n$.

\paragraph{Learning $\upvartheta$} Similar to the encoder, we can minimize decoder parameters $\upvartheta$ with respect to $\mathscr{W}_{d_I}\bracket{\hat{P}^n, p_{\upvartheta} \# \gamma}$. In order to compute this quantity, one may break it down using certain advances in optimal transport.

\begin{theorem}[\cite{tolstikhin2017wasserstein}]\label{thm:wae}
Let $P \in \Delta(\Omega)$ and for some $n \in \mathbb{N}$ and $\hat{P}^n$ correspond to the empirical distribution of $P$ from $n$ i.i.d samples. Let $\mathcal{Q} = \braces{\pi : \int_{\Omega} \pi(z \mid X) d\hat{P}^n(X) = \gamma(z)}$, we then have that 
    \begin{equation}
       \mathscr{W}_{d_I}\bracket{\hat{P}^n, p_{\upvartheta} \# \gamma} =\nonumber
       \hspace{2em} \inf_{\pi \in \mathcal{Q} } \int_{\Omega} \int_{\mathcal{Z}} d_{I}(X,D_{\upvartheta}(z))d\pi(z \mid X) d\hat{P}^n(X).
   \end{equation}
\end{theorem}
\normalsize
The main contribution of this Theorem is that instead of finding an optimal coupling, as is standard in Wasserstein distances, we can search for this conditional distribution $\pi$ instead by searching within $\mathcal{Q}$. Notice however that $\pi$ is precisely an encoder distribution mapping from $\mathcal{X}$ to a distribution over $\mathcal{Z}$ and the constraint of $\mathcal{Q}$ amounts to ensuring $\pi \# \hat{P}^n = \gamma$, which can be enforced by minimizing $d_{\operatorname{MMD}}(\pi \# \hat{P}^n, \gamma)$, which is precisely the same objective we derived for the encoder $q_{\xi}$. Hence we can replace the minimization over $\pi$ with $q_{\xi}$, leading to the overall loss function
\begin{equation}
     \mathsf{L}(\upvartheta, \xi):= 
     \underbrace{\int_{\mathcal{Z}} d_{I}(X,D_{\upvartheta}(z))dq_{\xi}(z \mid X) d\hat{P}^n(X)}_{\mathclap{\text{Reconstruction error}}} + \lambda \cdot \underbrace{d_{\operatorname{MMD}}\bracket{q_{\xi} \# \hat{P}^n, \gamma}}_{\mathclap{\text{Prior regularization}}},
    \label{glo_formula}
\end{equation}
where $\lambda$ can be interpreted as a Lagrangian multiplier for the constraint $q_{\xi} \# \hat{P}^n=  \gamma$. The first term can be approximated with Monte Carlo sampling and resembles a reconstruction error. In order to approximate $d_{I}$, we utilize the Sinkhorn distance with cost matrix $c_I$ which has demonstrated to be an efficient differentiable proxy to Wasserstein distances. The MMD term can also easily be computed, thus leading to a computationally viable algorithm.
\section{Experiments}\label{sec:exp}
In this section, we compare the proposed methods to validate our theoretical findings related to the benefits provided by the geometric formulation. We will focus our experiments on four different data sets: Movielens20M, Netflix and two internal datasets from a large scale streaming service. 

\subsection{Experimental setup} \label{subsec:stat}
\paragraph{Metrics and statistical significance} In order to evaluate the methods, we will utilize $\operatorname{Recall}@K$ and $\operatorname{nDCG}@K$ defined as follows.
$\operatorname{Recall}@K = \sum_{k=1}^K \frac{\ivb{r(k) \in I_1}}{\min(K, \card{I_1})}$ and 
    $\operatorname{DCG}@K = \sum_{k=1}^K \frac{2^{\ivb{r(k) \in I_1}} - 1}{\log(k+1)}$ 
where $\ivb{A}$ is the Iverson bracket for an event $A$ (denoting $1$ if $A$ is true and $0$ otherwise), $r(k)$ indicates the item at rank $k$ and $I_1$ denotes the positive items set for the user. We then average across all users for both recall and nDCG. Recall indicates if the first $K$ predicted elements match with ground truth whereas DCG measures how well the predicted ranking matches the ground truth ranking. $\operatorname{nDCG}@K$ is obtained by normalizing DCG by dividing by the maximum possible value. We also perform bootstrapping analysis to assess the significance of our results by sampling 20\% of the users in each test dataset and repeating the process 1000 times, in order to compute 95\% confidence intervals, and we indicate when our results achieve statistical significance.

\paragraph{Baselines} We will be considering the following baselines.
\begin{itemize}[leftmargin=*]

    \item \textbf{ItemKNNCF}~\cite{dacrema2019we} A classic item-based nearest-neighbor method that will utilize Cosine Similarity to find the closest items. We will pick neighborhood size from $[5,1000]$.

    \item \textbf{P3alpha}~\cite{cooper2014random}. P3alpha is a graph-based method that performs random walks between users and items based on the observed interaction matrix.
    \item \textbf{RP3beta}~\cite{paudel2016updatable} A variant of P3alpha where the similarity distance between two items is divided by the popularity and raised to the power of an additional parameter $\beta$. 

    \item \textbf{SLIM}~\cite{ning2011slim} SLIM is a linear latent factor model which solves a constrained linear optimization problem and finds an asymmetric item-similarity matrix via sparsity and $\ell_1$, $\ell_2$ regularization.

    \item \textbf{WMF}~\cite{hu2008collaborative} Weighted Matrix Factorization (WMF) is a generalized matrix factorization by adopting a Gaussian observation sampling distribution with a item-dependent diagonal variance matrix, where unobserved items are set to $1$. 

    \item \textbf{NeuMF}~\cite{dziugaite2015neural, he2017neural} A generalization of matrix factorization method that learns a non-linear interaction between the item and user interactions. As suggested in \cite{he2017neural}, we adopt a $3$-layer MLP along with the number of predictive factors selected from $\braces{8,16,32,64}$.

    \item \textbf{MultiVAE}~\cite{liang2018variational} MultiVAE is a successful deep generative model inspired by Variational Inference, which derives a latent model for collaborative filtering. The autoencoder architecture for MultiVAE uses $1$ hidden layer of dimensionality $600$ and has the following end-to-end structure:
    $
    \mathbb{R}^{\card{I}} \xrightarrow[]{} \R^{600}  \xrightarrow[]{} \R^{200} \xrightarrow[]{} \R^{600} \xrightarrow[]{} \mathbb{R}^{\card{I}}
    $.
    
    \item \textbf{SinkhornCF}~\cite{li2021sinkhorn} This is a method that proposes using a Sinkhorn loss in addition to MultiVAE with the same motivation as our approach. This work uses the same architecture as MultiVAE with $\varepsilon = 1.0$ set as the regularization parameter for the Sinkhorn algorithm, batch sizes of $500$ are used with $100$ epochs. This method also relies on a cost matrix to be pre-specified similar to our approach and estimates this via the cosine value.

\end{itemize}
\paragraph{Architecture} For GeoCF, we adopt the same architecture as MultVAE and SinkhornCF. We use $\mathbb{R}^{\card{I}} \xrightarrow[]{} \R^{600}  \xrightarrow[]{} \R^{200} \xrightarrow[]{} \R^{600} \xrightarrow[]{} \mathbb{R}^{\card{I}}.$
We use a batch size of $500$ and $100$ epochs, as for SinkhornCF. For the Maximum Mean Discrepancy (MMD) term, we use Gaussian RBF kernel and select the kernel bandwidth parameter within the range $\braces{0.05,0.1,0.5,1.0}$. We use an exponential scheduler over the epochs, decreasing the Lagrangian multiplier $\lambda$ and use $\varepsilon = 1.0$ as the regularization parameter for the Sinkhorn algorithm.

\paragraph{Infrastructure} Our approach and all the baselines were evaluated on a p3.2xlarge AWS EC2 instance.

\subsection{Public Datasets}
Table \ref{table:movielenz} and Table \ref{table:netflix} present the performance of different models in terms of nDCG and Recall on Movielens20M\footnote{https://grouplens.org/datasets/movielens/20m/} and Netflix\footnote{https://www.kaggle.com/netflix-inc/netflix-prize-data}, respectively. Movielens20M consists of 136,677 users and 20,108 movies, while Netflix contains 480,189 different users and 17,770 items, where each record contains a user, item pair followed by a score rating the user has for the given item. We follow the same process as in~\cite{li2021sinkhorn}, keeping users that have watched at least five movies and threshold items with at least four ratings as a positive interaction. We split the user data into a training, validation and test set. For every user in the training set, we utilize all interaction history however for users in the validation or test set, a fraction of the history (80\%) is used to predict the remaining interaction. 

\begin{table*}[]
\caption{Recall \& nDCG for MovieLens20M. Best result is made bold, and the runner-up is underlined. * indicates statistical significance.}
\label{table:movielenz}
\centering
\begin{tabular}{@{}lllllllll@{}}
\toprule
\multicolumn{1}{c}{\multirow{2}{*}{\textbf{Model}}} & \multicolumn{4}{c}{\textbf{Recall}}                               & \multicolumn{4}{c}{\textbf{nDCG}}                                 \\ \cmidrule(l){2-9} 
\multicolumn{1}{c}{}                                & @20            & @50            & @75            & @100           & @20            & @50            & @75            & @100           \\ \midrule
ItemKNNCF                                           & 0.311          & 0.435          & 0.504          & 0.555          & 0.266          & 0.306          & 0.329          & 0.346          \\
P3alpha                                             & 0.291          & 0.397          & 0.446          & 0.475          & 0.244          & 0.280          & 0.298          & 0.308          \\
RP3beta                                             & 0.340          & 0.463          & 0.528          & 0.571          & 0.290          & 0.331          & 0.353          & 0.368          \\
SLIM                                                & 0.375          & 0.498          & 0.566          & 0.615          & 0.328          & 0.367          & 0.391          & 0.407          \\
WMF                                                 & 0.361          & 0.493          & 0.561          & 0.611          & 0.307          & 0.351          & 0.376          & 0.392          \\
NeuMF                                               & 0.161          & 0.321          & 0.419          & 0.494          & 0.110          & 0.168          & 0.195          & 0.221          \\
MultiVAE           & 0.400          & 0.538          & 0.610          & 0.661          & 0.339          & 0.387          & 0.412          & 0.429          \\
SinkhornCF                                          & {\ul 0.405}    & {\ul 0.542}    & {\ul 0.613}    & {\ul 0.665}    & {\ul 0.349}    & {\ul 0.395}    & {\ul 0.420}    & {\ul 0.438}    \\
GeoCF                                               & \textbf{0.411}* & \textbf{0.550}* & \textbf{0.622}* & \textbf{0.673}* & \textbf{0.353}* & \textbf{0.400}* & \textbf{0.423}* & \textbf{0.443}* \\ \midrule
Improvement (\%)                                    & 1.48           & 1.47           & 1.46           & 1.20           & 1.14           & 1.26           & 0.71           & 1.14           \\ \bottomrule
\end{tabular}
\end{table*}

\begin{table*}[]
\caption{Recall \& nDCG for Netflix. Best result is made bold, and the runner-up is underlined. * indicates statistical significance.}
\label{table:netflix}
\centering
\begin{tabular}{@{}lllllllll@{}}
\toprule
\multicolumn{1}{c}{\multirow{2}{*}{\textbf{Model}}} & \multicolumn{4}{c}{\textbf{Recall}}                               & \multicolumn{4}{c}{\textbf{nDCG}}                                 \\ \cmidrule(l){2-9} 
\multicolumn{1}{c}{}                                & @20            & @50            & @75            & @100           & @20            & @50            & @75            & @100           \\ \midrule
ItemKNNCF                                           & 0.271          & 0.348          & 0.406          & 0.454          & 0.251          & 0.269          & 0.288          & 0.305          \\
P3alpha                                             & 0.259          & 0.298          & 0.321          & 0.337          & 0.242          & 0.250          & 0.259          & 0.266          \\
RP3beta                                             & 0.307          & 0.378          & 0.426          & 0.463          & 0.288          & 0.302          & 0.319          & 0.334          \\
SLIM                                                & 0.347          & 0.426          & 0.486          & 0.534          & 0.325          & 0.341          & 0.361          & 0.379          \\
WMF                                                 & 0.314          & 0.393          & 0.450          & 0.496          & 0.293          & 0.310          & 0.330          & 0.347          \\
NeuMF                                               & 0.107          & 0.213          & 0.299          & 0.373          & 0.089          & 0.130          & 0.162          & 0.190          \\
MultiVAE                                             & 0.350          & 0.441          & 0.504          & 0.555          & 0.321          & 0.344          & 0.366          & 0.385          \\
SinkhornCF                                          & {\ul 0.356}    & {\ul 0.445}    & {\ul 0.507}    & {\ul 0.558}    & {\ul 0.327}    & {\ul 0.349}    & {\ul 0.371}    & {\ul 0.389}    \\
GeoCF                                               & \textbf{0.360}* & \textbf{0.450}* & \textbf{0.511}* & \textbf{0.561} & \textbf{0.330}* & \textbf{0.351}* & \textbf{0.374}* & \textbf{0.391}* \\ \midrule
Improvement (\%)                                    & 1.12           & 1.12           & 0.79           & 0.54           & 0.53           & 0.57           & 0.81           & 0.53           \\ \bottomrule
\end{tabular}
\end{table*}

\begin{table*}
\caption{nDCG results on two internal datasets. * indicates statistical significance. Left: \textit{Storefront} service. Right: \textit{NextUp} service.}
\label{table:minitv}
\parbox{.42\linewidth}{
\centering
\begin{tabular}{@{}lllll@{}}
\toprule
\multicolumn{1}{c}{\multirow{2}{*}{\textbf{Model}}} & \multicolumn{4}{c}{\textbf{nDCG}}                                                                             \\ \cmidrule(l){2-5} 
\multicolumn{1}{c}{}                                & @10                       & @20                       & @50                       & @100                      \\ \midrule
ItemKNNCF                                           & 0.521                     & 0.573                     & 0.599                     & 0.611                     \\
P3alpha                                             & 0.497                     & 0.517                     & 0.538                     & 0.576                     \\
RP3beta                                             & 0.500                     & 0.522                     & 0.535                     & 0.541                     \\
SLIM                                                & 0.543                     & 0.563                     & 0.581                     & 0.588                     \\
WMF                                                 & 0.501                     & 0.517                     & 0.525                     & 0.530                     \\
NeuMF                                               & 0.221                     & 0.256                     & 0.277                     & 0.190                     \\
MultiVAE &0.522&0.562 &0.580 &0.590 \\
SinkhornCF                                          & {\ul 0.602}               & {\ul 0.613}               & {\ul 0.626}              & {\ul 0.628}               \\
GeoCF                                               & \textbf{0.613}*           & \textbf{0.631}*            & \textbf{0.632}*            & \textbf{0.632}*            \\ \midrule
+ (\%)                                    & 1.83                      & 2.94                      & 0.96                      & 0.64                      \\ \bottomrule
\end{tabular}}
\hspace{1.5cm}
\parbox{.42\linewidth}{
\centering
\begin{tabular}{@{}lllll@{}}
\toprule
\multicolumn{1}{c}{\multirow{2}{*}{\textbf{Model}}} & \multicolumn{4}{c}{\textbf{nDCG}}                                     \\ \cmidrule(l){2-5} 
\multicolumn{1}{c}{}                                & @10             & @20             & @50             & @100            \\ \midrule
ItemKNNCF                                           & 0.502           & 0.538           & 0.576           & 0.605           \\
P3alpha                                             & 0.484           & 0.500           & 0.518           & 0.532           \\
RP3beta                                             & 0.488           & 0.502           & 0.519           & 0.534           \\
SLIM                                                & 0.525           & 0.541           & 0.561           & 0.579           \\
WMF                                                 & 0.493           & 0.510           & 0.530           & 0.547           \\
NeuMF                                               & 0.178           & 0.260           & 0.324           & 0.380           \\
MultiVAE & 0.417          & 0.422          & 0.443          & 0.447          \\
SinkhornCF                                          & {\ul 0.558}    & {\ul 0.587}    & {\ul 0.602}    & {\ul 0.611}    \\
GeoCF                                               & \textbf{0.608}* & \textbf{0.628}* & \textbf{0.646}* & \textbf{0.649}* \\ \midrule
+ (\%)                                    & 8.87            & 6.86            & 7.21            & 6.17            \\ \bottomrule
\end{tabular}}
\end{table*}

We highlight the best results in bold and underline the runner-ups. We can see that GeoCF achieves the best performance across the datasets with SinkhornCF consistently following as the runner-up. We believe these results testify to the importance of using geometric losses given such consistency. In particular, we note that the benefits gained by GeoCF are slightly higher on MovieLens20M compared to Netflix.

\subsection{Internal Datasets}
We now apply GeoCF to two internal datasets, which come from a streaming service exposed to a number of customers that go beyond the scale of the public datasets considered above. The first internal dataset involves all user-item interactions on the \textit{Storefront} of the service. The second is related to the \textit{NextUp} auto-play service, which triggers next title in auto-play mode upon completion of the current title. For both datasets, we collect 9 weeks of streaming data, and use the following week for testing. The last week of the 9 weeks training period is used for validation. We filter customers with at least 1 click over the train dataset. 

We present the results left and right of Table \ref{table:minitv} respectively. Similarly to MovieLens-20M and Netflix datasets, we find that GeoCF and SinkhornCF appear dominant quantitatively. In particular, we see larger boosts on the \textit{NextUp} service, which goes with the intuition that there is a stronger semantic bias in auto-play recommendations. 


\section{Conclusion}
We propose a notion of generalization gap in collaborative filtering and analyze this with respect to latent collaborative filtering models. We present a geometric upper bound that gives rise to a loss function, and a way to meaningfully utilize the geometry of item-metadata to improve recommendations. The resulting algorithm has deep connections to similiar approaches in genertive modelling such as the Wasserstein Autoencoder and SinkhornCF, both of which advocate for the use of geometry. The proposed approach achieves an improvement on both public and internal datasets, establishing the practical significance of our approach and validating the theoretical findings. The benefits of our method can be further extended with a more deliberate approach to learning an item-similarity matrix in the form of a bi-level optimization. We leave such pursuits as the subject of future work.



\bibliography{main}
\bibliographystyle{plain}

\appendix
\section{Appendix}\label{sec:appendix}

\subsection{Proof of Generalization}
\label{sec:generalization}

\textbf{Proof of Theorem~\ref{main_result}.}
Let $\mathcal{F}_c$ denote the set of $1$-Lipschitz functions with respect to $c$ and $\mathcal{H}_k$ denote the unit ball of Reproducing Kernel Hilbert Space functions with kernel $k$. For any $P \in \Delta(\Omega)$ and $\lambda > 0$, consider the minimizers $
    \upvartheta^{*}, \xi^{*} = \arginf_{\upvartheta,\xi} \mathsf{L}(\upvartheta, \xi).$
It then holds that with high probability

\small
\begin{equation}
\begin{split}
    \mathcal{E}_{n,P}(C_{\upvartheta^{*}, \xi^{*}}) \leq &\delta(\upvartheta^{*}, \xi^{*}; P) \\
    &+ \frac{1}{4} \E\left[ \mathscr{W}_c\bracket{\hat{P}^n, p_{\upvartheta^{*}} \# \gamma}
    + d_{\operatorname{MMD}}\bracket{q_{\xi^{*}} \# 
    \hat{P}^n, \gamma}\right] \\&+ An^{-\frac{1}{d^{*}_c(P)}},
\end{split}
\end{equation}
\normalsize

for some $A > 0$ where $q_{\xi^{*}}\# \hat{P}^n (z)= \int_{\mathcal{X}} q_{\xi^{*}}(z \mid \X) d\hat{P}^n(\X)$, $p_{\upvartheta^{*}}\# \gamma = \int_{\mathcal{Z}} p_{\upvartheta^{*}}(\cdot \mid z) \gamma(z)$ and $\delta(\upvartheta^{*}, \xi^{*}; P) = 2 \inf_{h \in \mathcal{F}_c} \nrm{\log \frac{P}{p_{\upvartheta^{*}} \# \gamma} - h}_{\infty} + 2\inf_{h \in \mathcal{H}_k}\nrm{\log \frac{\gamma}{q_{\xi^{*}} \# P} - h}_{\infty}$.

\begin{proof}
Throughout, we will assume that $\mathcal{F}_c$ are the set of $1$-Lipschitz functions with respect to $c$. Additionally, we will use $d_{\mathcal{F}}$ to denote the Integral Probability Metric (IPM) with function class $\mathcal{F}$: $d_{\mathcal{F}}(P,Q) = \sup_{f \in \mathcal{F}}\braces{\E_{P}[f] - \E_{Q}[f]}$. First note that 
\begin{align}
\mathcal{E}_{n,P}(C_{\upvartheta^{*},\xi^{*}}) &= \E\left[d_{\operatorname{TV}}\bracket{P, \int_{\Omega} \int_{\mathcal{Z}} p_{\upvartheta^{*}}(\cdot \mid z) q_{\xi^{*}}(z \mid X) dP(\X) } \right]\\
&\leq \E\left[d_{\operatorname{TV}}\bracket{P, \int_{\mathcal{Z}} p_{\upvartheta^{*}}(\cdot \mid z)d \gamma(z) } \right] + \\ &\E\left[d_{\operatorname{TV}}\bracket{\int_{\mathcal{Z}} p_{\upvartheta^{*}}(\cdot \mid z)d \gamma(z), \int_{\Omega} \int_{\mathcal{Z}} p_{\upvartheta^{*}}(\cdot \mid z) q_{\xi^{*}}(z \mid X) dP(\X) } \right] \label{eq:second-term}
\end{align}
We now focus on bounding the first term. Using Pinskers inequality, we have for any $h \in \mathcal{F}_c$:
\begin{align}
    &4 d_{\operatorname{TV}}(P,p_{\upvartheta^{*}} \# \gamma)\\ &\leq d_{\operatorname{KL}}(P \;\|\; p_{\upvartheta^{*}} \# \gamma) + d_{\operatorname{KL}}( p_{\upvartheta^{*}} \# \gamma  \;\|\; P )\\
    &\leq \int_{\Omega} \log \frac{P(X)}{p_{\upvartheta^{*}} \# \gamma(X)} \cdot (P(X) - p_{\upvartheta^{*}} \# \gamma(X))dX\\
    &= \int_{\Omega} \bracket{\log \frac{P(X)}{p_{\upvartheta^{*}} \# \gamma(X)} - h(X)}\cdot (P(X) - p_{\upvartheta^{*}} \# \gamma(X))dX + \int_{\Omega} h(X) \cdot (P(X) - p_{\upvartheta^{*}} \# \gamma(X))dX \\
    &\leq \int_{\Omega} \bracket{\log \frac{P(X)}{p_{\upvartheta^{*}} \# \gamma(X)} - h(X)}\cdot (P(X) - p_{\upvartheta^{*}} \# \gamma(X))dX + d_{\mathcal{F}_c}(P,p_{\upvartheta^{*}} \# \gamma) \\
    &\leq \nrm{\log \frac{P}{p_{\upvartheta^{*}} \# \gamma} - h}_{\infty} \nrm{P - p_{\upvartheta^{*}} \# \gamma}_1 + d_{\mathcal{F}_c}(P,p_{\upvartheta^{*}} \# \gamma)\\
    &\leq 2 \nrm{\log \frac{P}{p_{\upvartheta^{*}} \# \gamma} - h}_{\infty} + d_{\mathcal{F}_c}(P,p_{\upvartheta^{*}} \# \gamma),
\end{align}
where $h$ is then taken to be the minimizer from $\mathcal{F}_c$. Next note that
\begin{align}
    d_{\mathcal{F}_c}(P,p_{\upvartheta^{*}} \# \gamma) \leq d_{\mathcal{F}_c}(P, \hat{P}^n) + d_{\mathcal{F}_c}(\hat{P}^n, p_{\upvartheta^{*}} \# \gamma).
\end{align}
Thus putting this together, we have
\begin{align}
    \E[d_{\operatorname{TV}}(P,p_{\upvartheta^{*}} \# \gamma)] &\leq \frac{1}{2} \inf_{h \in \mathcal{F}_c} \nrm{\log \frac{P}{p_{\upvartheta^{*}} \# \gamma} - h}_{\infty} + \frac{1}{4}\E[d_{\mathcal{F}_c}(P, \hat{P}^n)] + \frac{1}{4}\E[d_{\mathcal{F}_c}(\hat{P}^n, p_{\upvartheta^{*}} \# \gamma)]\\
    &\leq \frac{1}{2} \inf_{h \in \mathcal{F}_c} \nrm{\log \frac{P}{p_{\upvartheta^{*}} \# \gamma} - h}_{\infty} + O\bracket{n^{-\frac{1}{d^{*}(P)}}} + \frac{1}{4} \E\left[ d_{\mathcal{F}_c}(\hat{P}^n, p_{\upvartheta^{*}} \# \gamma)\right],
\end{align}
where we utilize \cite[Lemma~21]{husain2019primal} for the last step and the infinum is taken by definition of being a minimizer. Note that we have $d_{\mathcal{F}_c} = \mathscr{W}_c$ by the Kantorovich-Rubinstein duality. For the second term in Eq \eqref{eq:second-term}, we can decompose it by taking the dual formulation of total variation:
\begin{align}
    &\E\left[d_{\operatorname{TV}}\bracket{\int_{\mathcal{Z}} p_{\upvartheta^{*}}(\cdot \mid z)d \gamma(z), \int_{\Omega} \int_{\mathcal{Z}} p_{\upvartheta^{*}}(\cdot \mid z) q_{\xi^{*}}(z \mid X) dP(\X) } \right] \\ &=\E\left[\sup_{g : \nrm{g}_{\infty} \leq 1} \braces{ \int_{\Omega} \int_{\mathcal{Z}} g(\X) dp_{\upvartheta^{*}}(\X \mid z) d\gamma(z) - \int_{\Omega} \int_{\Omega} \int_{\mathcal{Z}} g(\X) dp_{\upvartheta^{*}}(\X \mid z) dq_{\xi^{*}}(z \mid \X)dP(\X)  }  \right]\\
    &\leq \E\left[ d_{\operatorname{TV}}(\gamma, q_{\xi^{*}} \# P) \right].
\end{align}
Let now $\mathcal{H}_k$ be the Reproducing Kernel Hilbert Space ball of unit $1$ \cite{gretton2012kernel} then we have by the same analysis above except for $h \in \mathcal{H}_k$:
\begin{align}
    &4 d_{\operatorname{TV}}(\gamma, q_{\xi^{*}} \# P )\\ &\leq d_{\operatorname{KL}}(\gamma \;\|\; q_{\xi^{*}} \# P) + d_{\operatorname{KL}}( q_{\xi^{*}} \# P  \;\|\; \gamma )\\
    &\leq \int_{\Omega} \log \frac{\gamma(z)}{q_{\xi^{*}} \# P(z)} \cdot (\gamma(z) - q_{\xi^{*}} \# P(z))dz\\
    &= \int_{\Omega} \bracket{\log \frac{\gamma(z)}{q_{\xi^{*}} \# P(z)} - h(z)}\cdot (\gamma(X) - q_{\xi^{*}} \# P(z))dz + \int_{\Omega} h(z) \cdot (\gamma(z) - q_{\xi^{*}} \# P(z))dz \\
    &\leq \int_{\Omega} \bracket{\log \frac{\gamma(z)}{q_{\xi^{*}} \# P(z)} - h(z)}\cdot (\gamma(z) - q_{\xi^{*}} \# P(z))dz + d_{\mathcal{H}_k}(\gamma,q_{\xi^{*}} \# P) \\
    &\leq \nrm{\log \frac{\gamma}{q_{\xi^{*}} \# P} - h}_{\infty} \nrm{\gamma - q_{\xi^{*}} \# P}_1 + d_{\mathcal{H}_k}(\gamma,q_{\xi^{*}} \# P)\\
    &\leq 2\nrm{\log \frac{\gamma}{q_{\xi^{*}} \# P} - h}_{\infty} + d_{\mathcal{H}_k}(\gamma,q_{\xi^{*}} \# P).
\end{align}
Utilizing again the triangle inequality we have
\begin{align}
    d_{\mathcal{H}_k}(\gamma,q_{\xi^{*}} \# P) \leq d_{\mathcal{H}_k}(\gamma,q_{\xi^{*}} \# \hat{P}^n) + d_{\mathcal{H}_k}(q_{\xi^{*}} \# \hat{P}^n,q_{\xi^{*}} \# P).
\end{align}
Note that the second term converges as a rate of $O(n^{-\frac{1}{2}})$ due to \cite{gretton2012kernel} and is thus dominated by the rate of convergence $O(n^{-\frac{1}{d^{*}_c(P)}})$. The first term here is the dual form of the MMD and thus combining this with the first bound concludes the proof.
\end{proof}

\subsection{Collaborative Filtering lower bound}
We begin this section by defining Besov smoothness.
\begin{definition}
\label{def:besov}
    For a function $f\in L^p(\Omega)$ for some $p\in (0,\infty]$, the $r$-th modulus of smoothness of $f$ is defined by
    \begin{align}
      &\quad  w_{r,p}(f,t) = \sup_{\|h\|_2\leq t}\|\Delta_h^r(f)\|_p,\quad\text{where }\ \Delta_h^r(f)(x)
        \\ &\hspace{-2mm} 
        = \begin{cases}
        \sum_{j=0}^r {r \choose j}(-1)^{r-j} f(x+jh) \hspace{-2mm} & (\text{if }x+jh\in \Omega\ \text{for all $j$})
        \\
        0 & (\text{otherwise}).
        \end{cases}
    \end{align}
\end{definition}
\begin{definition}[Besov space $B_{p,q}^s(\Omega)$]
    For $0<p,q\leq \infty, s>0, r:=\lfloor s\rfloor+1$, let the seminorm $|\cdot|_{B_{p,q}^s}$ be
    \begin{align}
        |f|_{B_{p,q}^s} = \begin{cases}\left(\int_0^\infty (t^{-s}w_{r,p}(f,t))^q \frac{dt}{t}\right)^\frac1q &(q<\infty ),
     \\   \sup_{t>0}t^{-s}w_{r,p}(f,t) &(q=\infty) .
        \end{cases}
    \end{align}
    The norm of the Besov space $B_{p,q}^s$ is defined by $\|f\|_{B_{p,q}^s} = \|f\|_p+|f|_{B_{p,q}^s}$, and we have $B_{p,q}^s=\{f\in L^p(\Omega)|\ \|f\|_{B_{p,q}^s}<\infty\}.$
\end{definition}
The Besov spaces constitute a rather large space of densities and is a standard assumption for the majority of density estimations works such as in \cite{}. These lines of work study minimax estimation rates for densities lying in Besov spaces, to which we prove a similar result for the case of collaborative filtering.

We now show, under an assumption that the target density satisfies Besov continuinity conditions.
\begin{proposition}
Let $P \in \Delta(\Omega)$. Let $0 < p,q \leq \infty, s > 0$ such that \begin{align}s > \max\bracket{d\bracket{\frac{1}{p} - \frac{1}{2}}, d\bracket{\frac{1}{p} - 1}, 0 }.\end{align} Then, we have that
\begin{align}
    \inf_{C: \Omega \to \Delta(\Omega)} \sup_{P \in B^{s}_{p,q}} \mathcal{E}_{n,P}(C) \gtrsim  n^{-\frac{s}{d + 2s}},
\end{align}
for some constant $c > 0$.
\end{proposition}
\begin{proof}
    This result is a consequence of Proposition 5.2 from \cite{oko2023diffusion}.
\end{proof}
The main takeaway from this result is that under smoothness assumptions (characterized by the parameter $s$), the best rate we can achieve depends on the dimensionality $d$, which arises precisely from the embedding functions we use to compare similarity between items. In particular, if we do not use any item similarity embedding, then the dimensionality $d$ will be the number of items in our collaborative filtering framework which can be quite large, leading to a larger error rate. Thus, this lower bound advocates for the use of encoders that reduce dimensionality.

\end{document}